%% file: ICSE-Main-189.tex
\documentclass[conference]{IEEEtran}
\IEEEoverridecommandlockouts
\usepackage{cite}
\usepackage{amsmath,amssymb,amsfonts,amsthm}
\usepackage{graphicx}
\usepackage{textcomp}
\usepackage{xcolor}

\usepackage{multicol}
\usepackage{multirow}
\usepackage{bigstrut}
\usepackage{xspace}
\usepackage{hyperref}
\usepackage{rotating}
\usepackage[linesnumbered,boxed,ruled,comment snumbered,vlined]{algorithm2e}
\usepackage{algpseudocode}
\usepackage{float}
\usepackage{threeparttable}
\usepackage{mathtools}
\usepackage{booktabs}
\usepackage{lscape}

\newcommand{\PPAc}{\texttt{PPFc}\xspace}
\newcommand{\JAc}{\texttt{JFc}\xspace}
\newcommand{\GA}{\texttt{GA}\xspace}
\newcommand{\TB}{\texttt{Tab}\xspace}

\newcommand{\X}{\bigcirc}
\newcommand{\nexttime}{\X}
\newcommand{\future}{\Diamond}
\newcommand{\globally}{\square}
\newcommand{\until}{\mathcal{U}}
\newcommand{\release}{\mathcal{R}}
\newcommand{\weakuntil}{\mathcal{W}}

\newcommand{\externalContrastyFilter}{\textsc{externalContrastyFilter}\xspace}

\newcommand{\BCs}{\mathcal{B}}
\newcommand{\CBCs}{\mathcal{B}_c}
\newcommand{\GBCs}{\mathcal{B}_g}

\newcommand{\ie}{{\it i.e.}}

\newcommand{\eg}{{\it e.g.}}
\newcommand{\St}{{\it s.t.}}
\newcommand{\resp}{{\it resp.}~}

\newcommand{\true}{\top}
\newcommand{\false}{\bot}
\newcommand{\LtlP}{\mathbb{P}}

\newtheorem{definition}{\bf Definition}
\newtheorem{example}{\bf Example}
\newtheorem{property}{\bf Property}
\newtheorem{thm}{\bf Theorem}

\newtheorem{que}{\bf{RQ}}

\newcommand{\funFont}[1]{\textsc{#1}\xspace}
\newcommand{\tabincell}[2]{\begin{tabular}{@{}#1@{}}#2\end{tabular}}

\def\BibTeX{{\rm B\kern-.05em{\sc i\kern-.025em b}\kern-.08em
    T\kern-.1667em\lower.7ex\hbox{E}\kern-.125emX}}
\begin{document}

\title{
    How to Identify Boundary Conditions with Contrasty Metric?
    \thanks{\IEEEauthorrefmark{2}Corresponding author.}
    \thanks{This paper was supported by the Guangdong Province Science and Technology Plan projects (No. 2017B010110011), National Natural Science Foundation of China (No. 61976232), National Key R\&D Program of China (No. 2018YFC0830600), Guangdong Province Natural Science Foundation (No. 2018A030313086 and 2017A070706010 (soft science)).}
}

\author{\IEEEauthorblockN{1\textsuperscript{st} Weilin Luo}
\IEEEauthorblockA{\textit{School of Computer Science and Engineering} \\
\textit{Sun Yat-sen University}\\
Guangzhou, China \\
luowlin3@mail2.sysu.edu.cn}
\and
\IEEEauthorblockN{2\textsuperscript{nd} Hai Wan\IEEEauthorrefmark{2}}
\IEEEauthorblockA{\textit{School of Computer Science and Engineering} \\
\textit{Sun Yat-sen University}\\
Guangzhou, China \\
wanhai@mail.sysu.edu.cn}
\and
\IEEEauthorblockN{3\textsuperscript{rd} Xiaotong Song}
\IEEEauthorblockA{\textit{School of Computer Science and Engineering} \\
\textit{Sun Yat-sen University}\\
Guangzhou, China \\
songxt5@mail2.sysu.edu.cn}
\and
\IEEEauthorblockN{4\textsuperscript{th} Binhao Yang}
\IEEEauthorblockA{\textit{School of Computer Science and Engineering} \\
\textit{Sun Yat-sen University}\\
Guangzhou, China \\
yangbh7@mail2.sysu.edu.cn}
\and
\IEEEauthorblockN{5\textsuperscript{th} Hongzhen Zhong}
\IEEEauthorblockA{\textit{School of Computer Science and Engineering} \\
\textit{Sun Yat-sen University}\\
Guangzhou, China \\
zhonghzh5@mail2.sysu.edu.cn}
\and
\IEEEauthorblockN{6\textsuperscript{th} Yin Chen}
\IEEEauthorblockA{\textit{School of Computer Science} \\
\textit{South China Normal University}\\
Guangzhou, China \\
ychen@scnu.edu.cn}
}

\maketitle

\begin{abstract}
The {\em boundary conditions (BCs)} have shown great potential in requirements engineering because a BC captures the particular combination of circumstances, \ie, {\em divergence}, in which the goals of the requirement cannot be satisfied as a whole.
Existing researches have attempted to automatically identify lots of BCs.
Unfortunately, a large number of identified BCs make assessing and resolving divergences expensive.
Existing methods adopt a coarse-grained metric, {\em generality}, to filter out less general BCs.
However, the results still retain a large number of redundant BCs since a general BC potentially captures {\em redundant} circumstances that do not lead to a divergence.
Furthermore, the {\em likelihood} of BC can be misled by redundant BCs resulting in costly repeatedly assessing and resolving divergences.

In this paper, we present a fine-grained metric to filter out the redundant BCs.
We first introduce the concept of {\em contrasty} of BC.
Intuitively, if two BCs are contrastive, they capture different divergences.
We argue that a set of contrastive BCs should be recommended to engineers, rather than a set of general BCs that potentially only indicates the same divergence.
Then we design a post-processing framework (\PPAc) to produce a set of contrastive BCs after identifying BCs.
Experimental results show that the contrasty metric dramatically reduces the number of BCs recommended to engineers.
Results also demonstrate that lots of BCs identified by the state-of-the-art method are redundant in most cases.
Besides, to improve efficiency, we propose a joint framework (\JAc) to interleave assessing based on the contrasty metric with identifying BCs.
The primary intuition behind \JAc is that it considers the search bias toward contrastive BCs during identifying BCs, thereby pruning the BCs capturing the same divergence.
Experiments confirm the improvements of \JAc in identifying contrastive BCs.
\end{abstract}


\section{Introduction}\label{sec:introduction}

{\em Goal-oriented requirement engineering (GORE)}~\cite{van2009requirements} is an essential phase of the software development life cycle,
the important task of which is to attain correct software requirements specifications.
Many researches have demonstrated the significant advantages that formal and goal-oriented approaches help generate correct specifications~\cite{alrajeh2009learning,degiovanni2014automated,ellen2014detecting}.
In such approaches, {\em domain properties} and {\em goals} are represented in {\em linear-time temporal logic (LTL)} because LTL is proved convenient for abstracting specifications of a large class of requirements, assumptions, and domain properties~\cite{van2009requirements}.

The identify-assess-control cycle in GORE aims at identifying, assessing, and resolving inconsistency in which the goals of the requirement cannot be satisfied as a whole.
The {\em divergence} is a weak inconsistency, \ie, particular circumstances where the satisfaction of some goals inhibits the satisfaction of others.
A divergence is captured by {\em boundary conditions (BCs)} which explain why the divergence happens.
Various approaches~\cite{van1998managing,degiovanni2016goal,degiovanni2018genetic} have been proposed to automatically identify BCs in the context of GORE.

As the number of identified BCs in the identification stage increases, for example, there are more than $100$ BCs in the case named London Ambulance Service in~\cite{degiovanni2018genetic}, the assessment stage and the resolution stage become very expensive, and even impractical.
In order to provide engineers with an acceptable number of BCs to analyze, the {\em generality} metric (Definition~\ref{def:gbc})~\cite{degiovanni2018genetic} has been proposed to automatically filter out the less general BCs.
The generality metric qualitatively distinguishes the importance of BC using the implication relationship of BCs.
Intuitively, a more general (also known as weaker) BC is more important because it potentially covers more circumstances to represent a divergence.
Therefore, the less general BCs can be filtered out by the more general one.

Unfortunately, we observe that a set of general BCs still retains a large number of redundant BCs.
The reason is that the generality metric can be considered as a coarse-grained metric.
A general BC potentially captures {\em redundant circumstances} that do not lead to a divergence.

Furthermore, the accuracy of the assessment step based on {\em likelihood} is sensitive to the redundant circumstances, so a set of general BCs can lead to mistakes in the assessment step (an example shown in Section~\ref{sec:motivation}).
The assessment stage is concerned with evaluating how likely the identified conflicts are, and how likely and severe are their consequences.
Degiovanni et al.~\cite{degiovanni2018goal} proposed an automatically assessing method based on model counting,
which can be used to prioritize BCs to be resolved.
However, a set of general BCs misleads to prioritize the BCs because a general BC potentially captures redundant circumstances that do not lead to a divergence.

In this paper, we present a new metric to assess the differences among the divergences captured by BCs.
Our approach is novel in the following respects:
(1) It is a fine-grained metric because it can filter out not only the less general BCs but also the BCs that capture the same divergence;
(2) and it measures the differences between BCs from the different divergences captured by them.
We first introduce the concept of {\em contrasty} of BCs motivated by avoiding boundary conditions~\cite{van1998managing} in resolving divergences.
More precisely, given two BCs $\phi$ and $\varphi$, we consider whether $\phi' = \phi \land \neg \varphi$ and $\varphi' = \varphi \land \neg \phi$ are BCs.
$\phi'$ (\resp $\varphi'$) represents the circumstances left by removing the circumstances captured by $\varphi$ (\resp $\phi$) from that captured by $\phi$ (\resp $\varphi$).
If neither $\phi'$ nor $\varphi'$ is BC, $\phi$ and $\varphi$ are contrastive.
Intuitively, if two BCs are contrastive, they capture different divergences.
We argue that a set of contrastive BCs should be recommended to engineers, rather than a set of general BCs since they potentially only indicate the same divergence.

Based on the contrasty metric, we design a post-processing framework (\PPAc) to produce a set of contrastive BCs after identifying BCs.
Experimental results show that the contrasty metric can filter out all the BCs that capture the same divergence, which dramatically reduces the number of BCs recommended to engineers.
Furthermore, experiments show that the BCs identified by the state-of-the-art method are not contrastive in most cases.
In other words, these BCs capture the same divergence, in which engineers only consider one BC to resolve a divergence while others are redundant.

In order to improve efficiency, we propose a joint framework (\JAc) to interleave assessing based on the contrasty metric with identifying BCs.
Specifically, when a BC is identified during the search, we add its negation as an additional constraint to domain properties.
The additional constraint makes the domain properties dynamically change so that it prevents the same circumstances from being identified as a BC again.
The insight behind this is that it produces the search bias towards the BCs that capture different divergences.
Besides, we propose a sufficient condition for the case where there not exist BCs.
It guarantees that if we resolve the divergences captured by the BCs in the set of contrastive BCs, there not exist divergences under the original domain properties and goals.
Experiments confirm the improvements of \JAc in identifying contrastive BCs.

Our main contributions are summarized as follows.
\begin{itemize}
    \item We present the novel contrasty metric to evaluate the differences between BCs 
    , which can filter out more redundant BCs that capture the same divergence.
    \item We design a post-processing framework (\PPAc) to produce a set of contrastive BCs.
    In order to improve efficiency, we also design a joint framework (\JAc) to capture different divergences during the search.
    \item Experiments show that the contrasty metric is better than the generality metric for filtering out redundant BCs.
\end{itemize}

\section{Background}\label{sec:background}

In this section, we introduce the background of goal-conflict analysis and linear-time temporal logic.
We briefly recall some basic notions for the rest of the paper.

\subsection{Goal-Conflict Analysis}\label{sec:background-gci}

In GORE~\cite{van2009requirements}, goals are prescriptive statements that the system must achieve, and domain properties are descriptive statements that capture the domain of the problem world.
In practice, it is unrealistic to require requirements specifications to be complete or all goals to be satisfiable, because inconsistencies may occur.
{\em Goal-conflict analysis}~\cite{van1998integrating,van2009requirements} deals with the inconsistencies via the following identify-assess-control cycle:
\begin{enumerate}
    \item the \textit{identification stage} is to identify a condition whose occurrence makes some inconsistencies;
    \item the \textit{assessment stage} is to assess and prioritize the identified inconsistencies according to their likelihood and severity;
    \item the \textit{resolution stage} is to resolve the identified inconsistencies by providing appropriate countermeasures.
\end{enumerate}

\noindent{\bf Goal-Conflict Identification.}
In this paper, we focus on a weak inconsistency, \ie, {\em divergence}.
A divergence essentially represents a {\em boundary condition (BC)} whose occurrence results in the loss of satisfaction of the goals, which makes the goal divergence~\cite{van1998managing}.
\begin{definition}\label{def:bc}
    Let $G = \{g_1, \dots, g_n\}$ be a set of goals and $Dom$ a set of domain properties.
    A {\em divergence} occurs within $Dom$ iff there exists a {\em boundary condition} $\varphi$ under $Dom$ and $G$ such that the following conditions hold:
    \begin{flalign*}
    &Dom \wedge G  \wedge \varphi \models \bot    \tag{\text{logical inconsistency}} \\
    &Dom \wedge  G_{-i} \wedge \varphi \not \models \bot \text{, for  each } \! 1 \leq i \! \leq \! n  \tag{\text{minimality}} \\
    & \neg G \not \equiv \varphi \tag{\text{non-triviality}}
    \end{flalign*}
    where $G=\bigwedge_{1\leq i \leq n}g_i$ and $G_{-i} = \bigwedge_{j\not =i}g_j$.
\end{definition}

Intuitively, a BC captures a particular combination of circumstances in which the goals cannot be satisfied as a whole.
The logical inconsistency property means the conjunction of goals becomes inconsistent when $\varphi$ holds.
The minimality property states that disregarding any of the goals no longer results in inconsistency.
The non-triviality property forbids a BC to be a trivial condition which is the negation of the conjunction of the goals.
Note that BCs are not false due to the minimality property.

Specifying software requirements in the LTL formulation allows us to employ automated LTL satisfiability solvers to check for the feasibility of the corresponding requirements.
With an efficient LTL satisfiability solver, we can automatically check if the generated candidate formulae are valid BCs or not by checking if they satisfy the properties.

In the identification stage, the {\em generality}~\cite{degiovanni2018genetic} metric has been proposed to reduce the redundant BCs.
It is defined as follows.

\begin{definition}\label{def:gbc}
    Let $S$ be a set of BCs.
    A BC $\varphi_i \in S$ is {\em more general} than another BC $\varphi_j \in S$ if $\varphi_j$ implies $\varphi_i$.
\end{definition}

Intuitively, a more general BC $\varphi$ captures all the particular combinations of circumstances captured by the less general BCs than $\varphi$.
Therefore, it is important to provide engineers with more general BCs.
As far as we know, the generality metric is the only metric to filter out BCs.

\noindent{\bf Goal-Conflict Assessment.}
In the assessment stage, in order to give engineers more guidance on which BCs need to get attention, probabilities~\cite{cailliau2012probabilistic} of their occurrence are considered as an important indicator.
For systems without extra probabilistic information, there is an approach~\cite{degiovanni2018goal} based on model counting to analyze the likelihood of BCs.
It is defined as follows.

\begin{definition}\label{def:likelihood}
    Let $\phi$ be a BC, $Dom$ domain properties, and $k$ a positive integer.
    The likelihood of $\phi$ is $L(\phi) = \frac{\#(Dom \cup \phi, k)}{\#(Dom, k)}$ where $\#(C,k)$ denotes that the total number of models bases of length $k$ satisfying constraints in $C$.
\end{definition}

Intuitively, the larger likelihood of a BC indicates that the divergence captured by the BC is more likely to happen.

\noindent{\bf Goal-Conflict Resolution.}
In the resolution stage, as the BCs malfunction the system when the system reaches the circumstances captured by BCs, the engineers need some strategies to resolve the divergences captured by the BCs. 

\begin{definition}
    Let $Dom$ be domain properties, $G$ goals, and $\phi$ a BC under $Dom$ and $G$. 
    {\em Resolving divergences} aims to modify $Dom$ and $G$ to get $Dom'$ and $G'$, so that $\phi$ under $Dom'$ and $G'$ does not fulfill at least one of the following constraints: 
    \begin{enumerate}
        \item $Dom' \wedge G' \wedge \varphi \models \bot$;
        \item $Dom' \wedge  G'_{-i} \wedge \varphi \not \models \bot \text{, for  each } \! 1 \leq i \! \leq \! n$;
        \item $\neg G' \not \equiv \varphi$.
    \end{enumerate}
\end{definition}

Intuitively, after resolving divergences, the circumstances captured by the BC do not happen under the new system expressed by updated domain properties and goals.
Van Lamsweerde et al.~\cite{van1998managing} proposed that generating reasonably updated domain properties and goals is an open problem because it requires a lot of experience. 
We will illustrate an example of resolving divergences in~\ref{sec:relatedwork}.
Therefore, a large number of identified BCs make the resolution stage very expensive.

A straightforward strategy can be adopted to avoid the circumstances captured by a BC.
The {\em avoid pattern}~\cite{van1998managing} is therefore introduced: $\globally (Dom \to \globally \neg B)$ where $B$ denotes a BC to be inhibited.

\subsection{Linear-Time Temporal Logic}\label{sec:background-ltl}

Linear-Time Temporal Logic (LTL)~\cite{pnueli1977temporal} is widely used to describe infinite behaviors of discrete systems, which is suitable for specifying software requirements~\cite{van1998managing}.
Throughout this paper, we use lower case letters (\eg, $p$, $h$) to denote propositions.
The syntax of LTL for a finite set of propositions $\LtlP$ includes the standard logical connectives ($\wedge$, $\vee$, $\neg$), $\mathbb{B}=\{\false, \true\}$, and temporal operators {\em next} ($\X$), {\em until} ($\until$).
$$
\varphi \coloneqq \false \mid \true \mid p \mid \varphi_1 \wedge \varphi_2 \mid \varphi_1 \vee \varphi_2 \mid \neg \varphi \mid \X \varphi \mid \varphi_1 \until \varphi_2
$$

Operator release ($\release$), eventually ($\future$), always ($\globally$), and weak-until ($\weakuntil$) are commonly used, and can be defined as
$\varphi_1\release\varphi_2 \coloneqq \neg(\neg \varphi_1\until \neg \varphi_2)$,
$\future \varphi \coloneqq \true \until \varphi$,
$\globally \varphi \coloneqq \neg(\true \until \neg \varphi)$, and
$\varphi_1 \weakuntil \varphi_2 \coloneqq \varphi_1 \until (\varphi_2 \lor \globally \varphi_1)$, respectively.
We use $|\varphi|$ to denote the {\em size} of the formula $\varphi$, \ie, the number of temporal operators, logical connectives, and literals in $\varphi$.


LTL formulae are interpreted over a {\em linear-time structure}.
A linear-time structure is a pair of $W = (S,\varepsilon)$ where $S$ is a state sequence and $\varepsilon: S \rightarrow 2^\LtlP$ is a function mapping each state $s_i$ to a set of propositions.
Let $W$ be a linear-time structure, $i \geq 1$ a position, and $\varphi_1$, $\varphi_2$ two LTL formulae.
The {\em satisfaction relation} $\models$ is defined as follows:
\begin{center}\begin{tabular}{lcl}
       $W,i \models p$ &iff& $p \in \varepsilon(s_i) \text{, where } p \in \LtlP $
       \\
       $W,i \models \neg \phi$ &iff& $W,i \not \models \phi$
       \\
       $W,i \models \varphi_1 \wedge \varphi_2$ &iff& $W,i \models \varphi_1 \textsf{ and } W,i \models \varphi_2$
       \\
       $W,i \models \nexttime \varphi$ &iff& $W,i+1 \models \varphi$
       \\
       $W,i \models \varphi_1 ~\until~ \varphi_2$ &iff& $\exists k \geq i \text{ s.t. } W,k \models \varphi_2 \text{ and }$\\
       && $\forall i \leq j < k, W,j \models \varphi_1$
\end{tabular}\end{center}

An LTL formula $\varphi$ is called {\em satisfiable} if and only if there is a linear-time structure (model) satisfying $\varphi$.
An LTL formula $\varphi$ {\em implies} an LTL formula $\varphi'$, noted $\varphi \rightarrow \varphi'$, 
if the models of $\varphi$ are also models of $\varphi'$.
The LTL satisfiability problem is to check whether an LTL formula is satisfiable, which is PSPACE-complete~\cite{sistla1985complexity}.
Recently, LTL satisfiability checkers based on different techniques have been developed.
Among these checkers, nuXmv~\cite{cavada2014nuxmv} and Aalta~\cite{li2015sat} have achieved better performance.

\section{Motivating Example}\label{sec:motivation}

In this section, we will illustrate the drawbacks of the generality metric through an example and discuss the insights behind the contrasty metric.
Below we illustrate an example, MinePump~\cite{kramer1983conic}.
\begin{example}\label{example}
    Consider a system to control a pump inside a mine. The main goal of the system is avoiding flood in the mine. The system has two sensors. One detects the high water level ($h$), the other detects methane in the environment ($m$). When the water level is high, the system should turn on the pump ($p$). When there is methane in the environment, the pump should be turned off. Domain property ($Dom$) and goals ($G$) are represented via the following LTL formulae.

    \noindent \textbf{Domain Property:}
    \begin{enumerate}
        \item \textbf{Name}: PumpEffect ($d_1$) \\
        \textbf{Description}: The pump is turned on for two time steps, then in the following one the water level is not high. \\
        \textbf{Formula}:
        $\globally((p \wedge \nexttime p) \rightarrow \nexttime(\nexttime \neg h))$
    \end{enumerate}
    \noindent \textbf{Goals:}
    \begin{enumerate}
        \item \textbf{Name}: NoFlooding ($g_1$) \\
        \textbf{Description}: When the water level is high, the system should turn on the pump.\\
        \textbf{Formula}: $\globally(h \rightarrow \nexttime(p))$
        \item \textbf{Name}: NoExplosion ($g_2$)\\
        \textbf{Description}: When there is methane in the environment, the pump should be turned off.\\
        \textbf{Formula}: $\globally(m \rightarrow \nexttime(\neg p))$
    \end{enumerate}
\end{example}

Although the specification is consistent, \ie, all domain properties and goals can simultaneously be satisfied, this specification exhibits some goal divergences.
One of the BCs is $\varphi_1 = \future(h \wedge m)$, which captures the circumstances where the high water level and the methane occur at the same time.
Under this situation, two goals are unsatisfiable simultaneously within domain property.

We also consider other two BCs $\varphi_2 = h \land m$ and $\varphi_3 = \future(h \land \lnot m \land p \land \nexttime(\lnot h \land \lnot p \lor h \land (m \lor \lnot p)))$.
$\varphi_2$ captures the circumstances that the water level is high and the methane occurs at the beginning.
Through equivalent transformation, we can obtain $\varphi_3 = \future( (h \land \lnot m \land p) \land \nexttime( (\neg h \land m \land \neg p) \lor (h \land m \land p) \lor (\neg h \land \neg m \land \neg p) \lor (h \land m \land \neg p) \lor (h \land \neg m \land \neg p) ) )$.
Clearly, $\varphi_3$ captures five circumstances, where, in the future, the system will migrate from the state where the high water level occurs, the methane does not occur, and the pump is turned on ($h \land \lnot m \land p$) to the state described as follows.
\begin{enumerate}
    \item the high water level does not occur, the methane occurs, and the pump is not turned on ($\neg h \land m \land \neg p$);
    \item the high water level and the methane occur and the pump is turned on ($h \land m \land p$);
    \item the high water level and the methane do not occur and the pump is not turned on ($\neg h \land \neg m \land \neg p$);
    \item the high water level and the methane occur and the pump is not turned on ($h \land m \land \neg p$);
    \item the high water level occurs, the methane does not occur, and the pump is not turned on ($h \land \neg m \land \neg p$).
\end{enumerate}

Existing methods can search for a large number of BCs.
It makes the assessment and resolution stages very expensive, and even impractical.
In order to provide engineers with an acceptable number of BCs to be analyzed, it is necessary to proposed metric to filter out the redundant BCs.

If we apply the generality metric, we filter out $\varphi_2$ because $\varphi_1$ is more general than $\varphi_2$.
However, the generality metric cannot evaluate $\varphi_1$ and $\varphi_3$ since the generality relationship between them does not hold.
In the assessment stage, if we compute the likelihood based on the method~\cite{degiovanni2018goal}, we can classify $\varphi_3$ as being more likely than $\varphi_1$ in the long term, and engineers should prioritize $\varphi_3$ in the search for mechanisms that would allow us to reduce the chances of reaching $\varphi_3$.

Unfortunately, the assessment method~\cite{degiovanni2018goal} lacks the accuracy to compute the likelihood by model counting because some models are meaningless, \ie, there does not exist the circumstances to lead the divergence in reality.
Considering the circumstances captured by $\varphi_3$, we observe that the circumstances (1), (3), (4), and (5) violate $g_1$.
Therefore, they cannot satisfy the minimality of BC, which means that they cannot capture the divergence in reality.
These circumstances are {\em redundant}, so $\varphi_3' = \future( (h \land \lnot m \land p) \land \nexttime( h \land p \land m ) )$ stands for the circumstances captured by $\varphi_3$.
We find that $\varphi_1$ is more likely than $\varphi_3'$ using the assessment method~\cite{degiovanni2018goal}, so $\varphi_1$ should be prioritized.
Situations like this show that the accuracy of the assessment method based on likelihood is sensitive to redundant circumstances.

In addition, we find that $\varphi_1$, $\varphi_2$, and $\varphi_3$ capture the same divergence, in which the high water level and the methane occur at the same time, \ie, the circumstance captured by $\varphi_1$.
It is very useful to identify the BC like $\varphi_1$ in resolving divergences.
Engineers only resolve $\varphi_1$ instead of resolving $\varphi_3$ first and then $\varphi_1$.
It avoids wasting computing resources caused by assessing and resolving redundant BCs.

In this paper, motivated by avoiding boundary conditions~\cite{van1998managing} in resolving divergence, we introduce the concept of witness (Definition~\ref{def:witness}) and contrasty (Definition~\ref{def:contrasty}) of BCs.
Intuitively, the witness of a BC indicates the cause of divergence.
If the two BCs are not mutual witnesses, then the two BCs are contrastive, \ie, they capture different divergences.
In this case, $\varphi_1$ and $\varphi_3$ are not contrastive because $\varphi_1$ is a witness of $\varphi_3$, but not vice versa, which means that the divergences captured by $\varphi_1$ are wider than that captured by $\varphi_3$.
Therefore, we recommend $\varphi_1$ to engineers and filter out $\varphi_3$.


\section{Identifying Boundary Conditions with Contrasty Metric}\label{sec:contrasty}

In this section, we first introduce the concept of contrasty of BCs.
Then, we design a post-processing framework to identify a set of contrastive BCs. 

\subsection{Contrasty}

We first introduce the concepts of {\em witness} and {\em contrasty}. 

\begin{definition}\label{def:witness}
    Let $f$ be an LTL formula and $\varphi$ a BC.
    $f$ is a {\em witness} of $\varphi$ iff $\varphi \land \neg f$ is not a BC.
\end{definition}

In the definition, motivated by avoiding boundary conditions~\cite{van1998managing}
in resolving divergences, we use a negative LTL formula to avoid some circumstances, \ie, resolving the divergence.
Therefore, the witness $f$ of a BC $\varphi$ indicates why $\varphi$ is a BC.
If $f$ is a BC, it means that the divergence captured by $\varphi$ is also captured by $f$.

\begin{definition}\label{def:contrasty}
    Let $\phi$ and $\varphi$ be BCs.
    $\phi$ and $\varphi$ are {\em contrastive}, iff $\phi$ is not a witness of $\varphi$ and $\varphi$ is not a witness of $\phi$.
\end{definition}

\begin{definition}
    Let $\CBCs$ be a set of BCs.
    $\CBCs$ is contrastive, iff $\forall \phi, \varphi \in \CBCs \land \phi \neq \varphi$, $\phi$ and $\varphi$ is contrastive.
\end{definition}

Intuitively, the contrastive BCs capture different divergences.
We use an example to illustrate the definition of witness and contrasty.

\begin{example}[Example~\ref{example} cont.]
    $\varphi_1 = \future(h \wedge m)$, $\varphi_2 = h \wedge m$, and $\varphi_3 = \future(h \land \lnot m \land p \land \nexttime(\lnot h \land \lnot p \lor h \land (m \lor \lnot p)))$.
    Because $\varphi_1 \land \neg \varphi_3$ is also a BC, \eg, it captures the circumstances where the high water level and the methane occur at the beginning, $\varphi_3$ is not a witness of $\varphi_1$.
    $\varphi_1$ is a witness of $\varphi_3$ since $\varphi_3 \land \lnot \varphi_1$ does not satisfy the minimality constraint of BC, \ie, $d_1 \land g_1 \land (\varphi_3 \land \neg \varphi_1)$ is unsatisfiable.
    Therefore, $\varphi_1$ and $\varphi_3$ are not contrastive.
    $\varphi_1$ is a witness of $\varphi_2$ and $\varphi_2$ is not a witness of $\varphi_1$, so $\varphi_1$ and $\varphi_2$ are not contrastive.  
    Intuitively, $\varphi_1$ is more important than $\varphi_2$ since the divergence captured by $\varphi_1$ is wider than that captured by $\varphi_2$ ($\varphi_2$ is a special case of $\varphi_1$).
    $\varphi_2$ and $\varphi_3$ are contrastive since they express the occurrence of the high water level and the methane in different situations.
\end{example}

Based on the definition, we have the following theorems.
These theorems indicate the highlight of the contrasty metric.


    


\begin{thm}\label{thm:gen-con}
    Let $\phi$ and $\varphi$ be BCs.
    If $\phi \to \varphi$, then $\varphi$ is a witness of $\phi$. 
\end{thm}

It is straightforward to prove Theorem~\ref{thm:gen-con} because $\phi \land \neg \varphi$ is unsatisfiable.
Because of Theorem~\ref{thm:gen-con}, we have Theorem~\ref{thm:finer-grained}.

\begin{thm}\label{thm:finer-grained}
    Let $\CBCs$ be a set of contrastive BCs.
    $\forall \phi,\varphi \in \CBCs \land \phi \neq \varphi$, $\phi \not \to \varphi \land \varphi \not \to \phi$.
\end{thm}

Theorem~\ref{thm:finer-grained} shows that there is not a general relation between any two BCs in a contrastive BC set, while there can be a witness relation between some two BCs in a general BC set.
According to Theorem~\ref{thm:finer-grained}, the contrasty metric can be regarded as a finer-grained metric than the generality metric because contrasty metric can filter out more redundant BCs than the generality metric. 
Let us recall Example~\ref{example}. 
$\{\varphi_1, \varphi_3\}$ is general, but not contrastive. 
If the contrasty metric is considered, then $\{\varphi_1\}$ is a contrastive.

\begin{property}\label{thm:con-resol}
    Let $\phi$ and $\varphi$ be BCs. 
    If $\phi$ is a witness of $\varphi$ and $\varphi$ is not a witness of $\phi$, then resolving the divergence captured by $\phi$ leads to resolving the divergence captured by $\varphi$.
\end{property}


Property~\ref{thm:con-resol} shows that it is reasonable that engineers prioritize $\phi$ to resolve since the circumstances captured by $\phi$ include the circumstances captured by $\varphi$.

\begin{thm}\label{thm:can-capt-diff-dive}
    Let $\phi$ and $\varphi$ be two BCs. 
    If $\phi$ and $\varphi$ are contrastive, then $\phi$ and $\varphi$ capture different divergences.
\end{thm}

\begin{proof}[Sketch of proof]
    $\phi$ and $\varphi$ are contrastive, so $\phi$ (\resp $\varphi$) is not the witness of $\varphi$ (\resp $\phi$), which means that $\phi \land \lnot \varphi$ (\resp $\varphi \land \lnot \phi$) is still a BC. 
    The primary intuition behind $\phi \land \lnot \varphi$ (\resp $\varphi \land \lnot \phi$) is that after resolving the divergences captured by $\varphi$ (\resp $\phi$), there are still divergences captured by $\phi$ (\resp $\varphi$).
    Therefore, $\phi$ and $\varphi$ capture different divergences.
\end{proof}

Theorem~\ref{thm:can-capt-diff-dive} shows that contrastive BCs capture different divergences.
Therefore, it is meaningful to recommend a set of contrastive BCs to engineers.

According to the above analysis, we argue that a set of contrastive BCs should be recommended to engineers, rather than a set of general BCs since they potentially only indicate the same divergences.
In Section~\ref{sec:exp-adv}, we will discuss the different divergences captured by contrastive BCs and report the advantage of the contrasty metric.

\subsection{Post-Processing Framework}\label{sec:PPAc}

We design a \underline{p}ost-\underline{p}rocessing framework for \underline{f}iltering the BCs based on the \underline{c}ontrasty metric (\PPAc).
It takes a set of BCs ($\BCs$) identified by a BC solver
as inputs. 
Its output is a set of contrastive BCs ($\CBCs$). 

\begin{algorithm}[h]
    \caption{\PPAc}\label{alg:PPAc}
    \KwIn{a set of BCs $\BCs$.}
    \KwOut{a set of contrastive BCs $\CBCs$.}

    $\CBCs \gets \BCs$;\\

    \For{each BC $\phi \in \CBCs$}
    {\label{alg:PPAc-selBC}
        $\CBCs' \gets \CBCs / \phi$;\\
        $isContrastive, W \gets$ \Call{externalContrastyFilter}{$\phi$, $\CBCs'$};\label{alg:PPAc-rela}\\
        \If{$isContrastive$}
        {
            $\CBCs \gets \CBCs / W$;\label{alg:PPAc-remove}\\
        }
        \Else
        {
            $\CBCs \gets \CBCs / \phi$;\label{alg:PPAc-remove}\\
        }
    }
    
    \Return $\CBCs$; \\
\end{algorithm}



        

\begin{algorithm}[h]
    \caption{\funFont{externalContrastyFilter}}\label{alg:excontFil}
    \KwIn{a BC $\phi$ and a set of BCs $\BCs$.}
    \KwOut{whether $\phi$ is contrastive in $\BCs$ and a set of BCs $W$ filtered by $\phi$.}

    $W \gets \emptyset$;\\
    \For{each BC $\varphi \in \BCs$}
    {
        \If{$\phi$ is a witness of $\varphi$ and $\varphi$ is a witness of $\phi$}
        {\label{alg:excontFil-small}
            \If{the size of $\phi$ is larger then that of $\varphi$}
            {
                \Return{False, $\emptyset$};\\
            } 
            \Else{
                $W \gets W \cup \{ \varphi \}$;\\
            }
        }
        \ElseIf{$\phi$ is a witness of $\varphi$ and $\varphi$ is not a witness of $\phi$}{
            \label{alg:excontFil-wit}
            $W \gets W \cup \{ \varphi \}$;\\
        }
        \ElseIf{$\phi$ is not a witness of $\varphi$ and $\varphi$ is a witness of $\phi$}{
            \label{alg:excontFil-notwit}
            \Return{False, $\emptyset$};\\
        }
    }
    \Return{True, $W$};\\
\end{algorithm}

The pseudo code is outlined in Algorithm~\ref{alg:PPAc}.
At each iteration, we choose a BC $\phi \in \CBCs$ (Alg.~\ref{alg:PPAc} of line~\ref{alg:PPAc-selBC}), then discuss its relationship with other BCs $\varphi$ in $\CBCs$ (Alg.~\ref{alg:PPAc} of line~\ref{alg:PPAc-rela}).
If $\phi$ and $\varphi$ are witnesses of each other (Alg.~\ref{alg:excontFil} of line~\ref{alg:excontFil-small}), which means that $\phi$ and $\varphi$ capture the same divergences, we select the one with smaller size\footnote{The BC with smaller size is more compact, and easier to interpret.} to stay in $\CBCs$.
If $\phi$ is a witness of $\varphi$ and $\varphi$ is not a witness of $\phi$ (Alg.~\ref{alg:excontFil} of line~\ref{alg:excontFil-wit}), which means that the divergences captured by $\phi$ is wider than that captured by $\varphi$, we retain $\phi$; otherwise (Alg.~\ref{alg:excontFil} of line~\ref{alg:excontFil-notwit}), we remove $\phi$.
If $\phi$ and $\varphi$ are not witnesses of each other, we do not delete any one because they are contrastive. 

\begin{thm}\label{thm:PPAc-sou}
    When Algorithm~\ref{alg:PPAc} terminates, $\CBCs$ is contrastive.
\end{thm}

    

It is straightforward to prove Theorem~\ref{thm:PPAc-sou}.
Theorem~\ref{thm:PPAc-sou} guarantees that Algorithm~\ref{alg:PPAc} returns a set of contrastive BCs.
We illustrate our method through a running example as follows.

\begin{example}[Example~\ref{example} cont.]
    Assume that the BC solver returns the set of BC $\BCs = \{ \varphi_1, \varphi_2, \varphi_3 \}$, where $\varphi_1 = \future(h \wedge m)$, $\varphi_2 = h \wedge m$, and $\varphi_3 = \future(h \land \lnot m \land p \land \nexttime(\lnot h \land \lnot p \lor h \land (m \lor \lnot p)))$.
    $\CBCs$ is initialized to $\{ \varphi_1, \varphi_2, \varphi_3 \}$.
    At the first iteration, assume that \PPAc chooses $\varphi_2$. 
    $\varphi_2$ will be compared with $\varphi_1$ and $\varphi_3$.
    Because $\varphi_2$ is not a witness of $\varphi_1$ and $\varphi_1$ is a witness of $\varphi_2$, \externalContrastyFilter returns False and an empty set.
    $\CBCs$ will be updated as $\{ \varphi_1, \varphi_3 \}$.
    At the second iteration, assume that \PPAc chooses $\varphi_1$.
    $\varphi_1$ will be compared with $\varphi_3$.
    Because $\varphi_1$ is a witness of $\varphi_3$ and $\varphi_3$ is not a witness of $\varphi_1$, \externalContrastyFilter returns True and $\{ \varphi_3 \}$.
    $\CBCs$ will be updated as $\{ \varphi_1 \}$.
    Then \PPAc returns $\{ \varphi_1 \}$ and terminates.
\end{example}

\subsection{Discussion about completeness and Performance}

In this paper, we are not concerned with the completeness of identifying contrastive BCs, \ie, the divergences captured by contrastive BCs cover all the divergences captured by BCs that have been found. 
The reason is as follows.

Firstly, we focus on filtering out redundant BCs for better resolving divergences which is the fundamental purpose of GORE. 
In general, the better the identification result is, the easier the resolution stage is. 
Therefore, we argue that the identified BCs should be conducive to resolving divergences as much as possible rather than completeness.

Furthermore, the completeness of the BC set does not help to resolve divergences. 
Resolving divergences is a dynamic process.
After resolving a BC, some BCs in the original BC set are no longer BCs under the updated domain properties and goals. 
In this way, for resolving divergences, it is meaningless to get the complete BC set in the BC identification stage.

For example, a set of general BCs fulfills the completeness, but it still retains a large number of redundant BCs that capture the same divergences, so that engineers will do a lot of meaningless work for resolving divergences. 
In other words, although the generality metric satisfies the completeness, it will also increase the burden of resolving divergences.
By comparison, the contrasty metric first considers the optimization of BC resolving divergences.

\PPAc only begins to filter out redundant BCs after the BC solver returns a set of BCs.
A natural idea is to directly identify contrastive BCs during searching for BCs.
In this way, pruning the BCs capturing the same divergence can be performed directly in the search process, thereby speeding up the searching process.
Based on this idea, we will discuss a joint framework for identifying contrastive BCs in Section~\ref{sec:joint}.

\section{Evaluation of Contrasty}\label{sec:exp-adv}

In this section, we reported the advantage of the contrasty metric.
Here, we presented the first research question.

\begin{que}\label{RQ1}
    Compared with the generality metric, what are the advantages of the contrasty metric?
\end{que}

Given a set of BCs $\BCs$ identified by a BC solver, we applied different metrics to filter out redundant BCs.
For the competitor, we combined the generality metric and the likelihood to filter and sort the BCs.
Specifically, we first filtered out the less general BCs to produce a set of general BCs $\GBCs$ and then sorted them according to the likelihood of BC from high to low. 
Based on \PPAc, we computed a set of contrastive BCs $\CBCs$ and sorted them by the likelihood.
We analyzed the shortcomings of the generality metric and reported the advantages of the contrasty metric by comparing $\GBCs$ and $\CBCs$.

\begin{table}[t]
    \centering
    \caption{The details of cases}\label{tab:case}
    \scalebox{1}{
      \begin{tabular}{c|r|r|r|r}
      \toprule
      Case & \multicolumn{1}{c|}{\#Dom} & \multicolumn{1}{c|}{\#Goal} & \multicolumn{1}{c|}{\#Var} & \multicolumn{1}{c}{Size} \\
      \midrule
      RetractionPattern1 (RP1) & 0     & 2     & 2     & 9 \\
      RetractionPattern2 (RP2) & 0     & 2     & 4     & 10 \\
      Elevator (Ele) & 1     & 1     & 3     & 10 \\
      TCP   & 0     & 2     & 3     & 14 \\
      AchieveAvoidPattern (AAP) & 1     & 2     & 4     & 15 \\
      MinePump (MP)& 1     & 2     & 3     & 21 \\
      ATM   & 1     & 2     & 3     & 22 \\
      Rail Road Crossing System (RRCS)  & 2     & 2     & 5     & 22 \\
      Telephone (Tel) & 3     & 2     & 4     & 31 \\
      London Ambulance Service (LAS)   & 0     & 5     & 7     & 32 \\
      Prioritized Arbiter (PA) & 6     & 1     & 6     & 57 \\
      Round Robin Arbiter (RRA) & 6     & 3     & 4     & 77 \\
      Simple Arbiter (SA) & 4     & 3     & 6     & 84 \\
      Load Balancer (LB) & 3     & 7     & 5     & 85 \\
      LiftController (LC) & 7     & 8     & 6     & 124 \\
      \tabincell{c}{ARM's Advanced Microcontroller\\ Bus Architecture (AMBA)}  & 6     & 21    & 16    & 415 \\
      \bottomrule
      \end{tabular}%
    }
    \label{tab:addlabel}%
\end{table}%

\begin{landscape}
    \input{tab_eva}
\end{landscape}

\subsection{Benchmarks}

We evaluated contrasty on $16$ different cases introduced by~\cite{degiovanni2018genetic}.
The details of each case are shown in Table~\ref{tab:case} including the numbers of domain properties (column `\#Dom'), goals (column `\#Goal'), variables (column `\#Var'), and the total size of all formulae (column `Size') for the specification of each case.
The order of the cases is sorted by the size of all formulae from small to large.

\subsection{Experimental Setups}

We used the following experimental setups.
\begin{itemize}
    \item We employed the state-of-the-art BC solver\footnote{\url{http://dc.exa.unrc.edu.ar/staff/rdegiovanni/ASE2018.html}}~\cite{degiovanni2018genetic} denoted by \GA to identify BCs. 
    It is based on a genetic algorithm to search BCs.
    \item We followed the configuration of \GA described in the paper~\cite{degiovanni2018genetic} including the size of the initial population generated from such a specification and the limit of $50$ generations, \ie, $50$ evolutions of the genetic algorithm population.
    \item We invoked Aalta~\cite{li2015sat} as the LTL satisfiability checker to check whether an LTL formula is a BC, whether one BC is more general than the other, and whether one BC is a witness of the other.
    Note that \GA~\cite{degiovanni2018genetic} also used Aalta as the LTL satisfiability checker.
    \item We computed the likelihood of a BC by the method~\cite{degiovanni2018goal}. And we set $k$ to $1000$, which is used in the paper~\cite{degiovanni2018goal} for good accuracy.
    \item For each case, we ran the algorithm $10$ times and reported the mean data.
    \item All the experiments were run on the $2.13$GHz Intel E$7$-$4830$, with $128$ GB memory under GNU/Linux (Ubuntu $16.04$). 
\end{itemize}

\subsection{Experimental Results}

Table~\ref{tab:numBC} summarizes the number of BC in $\BCs$ (`$|\BCs|$'), $\GBCs$ (`$|\GBCs|$'), and $\CBCs$ (`$|\CBCs|$'), where the column `\#suc.' means the number of successful runs (out of 10 runs).
If \GA fails in all 10 runs, the results are marked by `N/A'. 
Overall, our method can solve all the cases that can be solved by \GA to identify BCs.
Clearly, if the solver cannot identify BCs, our method cannot perform the post-processing.

\begin{table}[htbp]
    \centering
    
    \caption{The number of BC recommended by different metrics}\label{tab:numBC}%
    \scalebox{1}{
      \begin{tabular}{c|r|r|r|r}
      \toprule
      Case  & \multicolumn{1}{c|}{$|\BCs|$} & \multicolumn{1}{c|}{$|\GBCs|$} & \multicolumn{1}{c|}{$|\CBCs|$} & \multicolumn{1}{c}{\#suc.} \\
      \midrule
      RP1   & 37.1  & 3.2   & \textbf{1.2} & 10 \\
      RP2   & 35.1  & 2.6   & \textbf{1.2} & 10 \\
      Ele   & 28    & 3.2   & \textbf{2.6} & 10 \\
      TCP   & 53.9  & 2.1   & \textbf{1.5} & 10 \\
      AAP   & 50.3  & 3.7   & \textbf{1.8} & 10 \\
      MP    & 40.7  & 4.5   & \textbf{1.4} & 10 \\
      ATM   & 64.4  & 3.4   & \textbf{1.2} & 10 \\
      RRCS  & 27.9  & 3     & \textbf{1} & 10 \\
      Tel   & 36.5  & 3 & \textbf{1} & 2 \\
      LAS   & N/A   & N/A   & N/A   & N/A \\
      PA    & N/A   & N/A   & N/A   & N/A \\
      RRA   & 40.571 & 3.14 & \textbf{1} & 7 \\
      SA    & N/A   & N/A   & N/A   & N/A \\
      LB    & N/A   & N/A   & N/A   & N/A \\
      LC    & N/A   & N/A   & N/A   & N/A \\
      AMBA  & N/A   & N/A   & N/A   & N/A \\
      \bottomrule
      \end{tabular}%
    }
\end{table}%

For most cases, \GA returns a large number of BCs thanks to the development of search-based methods. 
Note that such a large set of BC can cause a huge burden in the assessment stage and the resolution stage.
Seeing the columns `$|\GBCs|$' and `$|\CBCs|$', we observe that the size of $\CBCs$ is much smaller than that of $\GBCs$ for all cases.
It means that, compared with the generality metric, the contrasty metric can considerably reduce the number of BCs to be analyzed by engineers.


Table~\ref{tab:BC} summarizes the results of the different metrics, for the BCs identified for each of the case studies.
We selected the data that \GA got the most number of BC from $10$ times experiments for display.
The column `GL' (\resp `CL') illustrates the BCs (`BC') in $\GBCs$ (\resp $\CBCs$) and their rank (`Rank') based on the likelihood metric.
We also use the column `Rank' as the identification of BCs.
For every BCs $\phi$ in $\CBCs$, we report which BCs in $\GBCs$ (`Witness') $\phi$ is a witness of and the identification of $\phi$ in $\GBCs$ is marked in red.

For all cases, $\CBCs$ is much smaller than $\GBCs$ and $\CBCs$ is a subset of $\GBCs$, which confirms that the contrasty metric is a more finer-grained metric than the generality metric.
The results also show that a set of general BCs still retains the BCs that represent the same divergence.
Particularly, for MP, ATM, and RRA, the redundant BCs are too much to assess and resolve divergences efficiently.

From the column `Witness', every BC in $\GBCs$ can find a witness of it in $\CBCs$. 
This observation means that the BCs in $\CBCs$ capture all the divergences captured by the BCs in $\GBCs$.
Therefore, engineers only need to consider the BCs in $\CBCs$ when resolving divergences.
In addition, we also observe that the contrastive BCs rank lower in $\GBCs$ in Ele, TCP, AAP, MP, RRCS, TEL, and RRA. 
The reason, as mentioned above, is that the circumstances that cannot describe the divergence lead to mistakes of likelihood.
Such mistakes are serious, which will prevent engineers from grasping the main cause of the divergence quickly. 
It leads to costly assessing and resolving the same divergence repeatedly.

In summary, the generality metric cannot capture the difference between BCs. 
Surprisingly, lots of BCs identified by the state-of-the-art BC solver are redundant in most cases. 
It puts an expensive burden on assessing and resolving divergences.
The method we propose can compare this well and give a recommendation that is more conducive to saving the costs of assessing and resolving divergences.

\section{Joint Framework}\label{sec:joint}

In this section, we design a \underline{j}oint framework to interleave \underline{f}iltering based on the \underline{c}ontrasty metric with identifying BCs (\JAc).
We first introduce the termination condition for identifying BCs and then propose \JAc.

Motivated by the blocking clause approach to solving All-SAT problem~\cite{mcmillan02}, we consider excluding the circumstances captured by identified BCs in the search process to generate a search bias towards the BCs that capture different divergences.
Specifically, in the process of searching for BCs, once a BC $\phi$ is identified, we add $\neg \phi$ as an additional constraint to domain properties.
The additional constraint makes the domain properties dynamically change so that it can prevent the same circumstances from being identified as a BC again (Theorem~\ref{thm:isNotWit}).
Moreover, we will prove that the BCs under the additional constraint are also BCs under the original domain properties and goals (Theorem~\ref{thm:isOriBC}).

Before introducing \JAc, We first propose a sufficient condition for the case where there does not exist a BC (called {\em BC termination condition}).
    
\begin{thm}\label{thm:endBC}
    Let $Dom$ be domain properties and $G$ goals.
    If $\exists 1 \leq i \leq |G|, Dom \land G_{-i} \land \neg G_i \models \bot$, then there does not exist a BC under $Dom$ and $G$.
\end{thm}

\begin{proof}[Sketch of proof]
    We prove that if there exists a BC, then $\forall 1 \leq i \leq |G|, Dom \land G_{-i} \land \neg G_i \not \models \bot$.
    If there is a BC $\phi$ under $Dom$ and $G$, then $Dom \land G \land \phi \models \bot$ (logical inconsistency) and $\forall 1 \leq i \leq |G|, Dom \land G_{-i} \land \phi \not \models \bot$ (minimality).
    Because of the logical inconsistency, we have $\phi \to \neg (Dom \land G)$.
    Therefore, $Dom \land G_{-i} \land \phi \to Dom \land G_{-i} \land \neg (Dom \land G)$.
    Consider the minimality, we have $\forall 1 \leq i \leq |G|, Dom \land G_{-i} \land \neg (Dom \land G) \not \models \bot$, \ie, $\forall 1 \leq i \leq |G|, Dom \land G_{-i} \land \neg G_i \not \models \bot$.
\end{proof}

Based on Theorem~\ref{thm:endBC}, we can check whether there still exists a BC under the dynamical domain properties and goals.

\begin{algorithm}[h]
    \caption{\JAc}\label{alg:JAc}
    \KwIn{domain properties $Dom$ and goals $G$.}
    \KwOut{a set of contrastive BCs $\CBCs$.}

    $\CBCs \gets \emptyset$;\\

    \While{True}
    {
        $isEnd, \phi \gets$\Call{callBCSolver}{$Dom \cup \{ \neg \varphi | \varphi \in \CBCs \}$, $G$};\label{alg:JAc-callsolver}\\
        \If{$isEnd$}
        {
            \Return{$\CBCs$};\label{alg:JAc-solverEnd}\\
        }
        \Else
        {
            $W \gets$\Call{internalContrastyFilter}{$\phi$, $\CBCs$};\label{alg:JAc-inConFil}\\
            $\CBCs \gets \CBCs / W$;\label{alg:JAc-remove}\\
            \If{there is not a BC under $Dom \cup \{ \neg \varphi | \varphi \in \CBCs \}$ and $G$}
            {\label{alg:JAc-isEnd}
                \Return{$\CBCs$};\\
            }
        }
    }
\end{algorithm}

\begin{algorithm}[h]
    \caption{\funFont{internalContrastyFilter}}\label{alg:incontFil}
    \KwIn{a BC $\phi$ and a set of BCs $\BCs$.}
    \KwOut{a set of BCs $W$ filtered by $\phi$.}

    $W \gets \emptyset$;\\
    \For{each BC $\varphi \in \BCs$}
    {
        \If{$\phi$ is a witness of $\varphi$}{
            \label{alg:incontFil-wit}
            $W \gets W \cup \{ \varphi \}$;\\
        }
    }
    \Return{True, $W$};\\
\end{algorithm}

\JAc takes the domain properties $Dom$ and goals $G$ as inputs. 
Its output is a set of contrastive BCs $\CBCs$. 
The pseudo code is outlined in Algorithm~\ref{alg:JAc}.
In order to identify BCs, we involve existing BC solvers, \eg, \GA~\cite{degiovanni2018genetic} and \TB~\cite{degiovanni2016goal} (Alg.~\ref{alg:JAc} of line~\ref{alg:JAc-callsolver}).
Note that we consider the dynamical domain properties ($Dom \cup \{ \neg \varphi | \varphi \in \CBCs \}$).
If the BC solver terminates, we return $\CBCs$ (Alg.~\ref{alg:JAc} of line~\ref{alg:JAc-solverEnd}). 
Otherwise, unlike \PPAc, we update $\CBCs$ when identifying a new BC (Alg.~\ref{alg:JAc} of line~\ref{alg:JAc-inConFil}-\ref{alg:JAc-remove}).
Note that we only remove the BCs which the new BC is a witness of (Alg.~\ref{alg:incontFil} of line~\ref{alg:incontFil-wit}) because none of the BCs in $\CBCs$ is a witness of the new BC (Theorem~\ref{thm:isNotWit}).
Afterward, if there still exists a BC under $Dom \cup \{ \neg \varphi | \varphi \in \CBCs \}$ and $G$, we continue to involve BC solver; otherwise, return $\CBCs$ (Alg.~\ref{alg:JAc} of line~\ref{alg:JAc-isEnd}).

\begin{thm}\label{thm:isOriBC}
    Let $Dom$ be domain properties, $G$ goals, and $\BCs$ a set of BCs that has been identified.
    A LTL formula $\phi$ is a BC under $Dom$ and $G$, if $\phi$ is a BC under $Dom \cup \{ \neg \varphi | \varphi \in \BCs \}$ and $G$.
\end{thm}

\begin{proof}[Sketch of proof]
    Because $\forall \varphi \in \BCs$ is a BC under $Dom$ and $G$, we have $Dom \land (\bigwedge_{\varphi \in \BCs} \neg \varphi) \land G \land \phi \equiv Dom \land G \land \phi$. 
    Therefore, $Dom \land (\bigwedge_{\varphi \in \BCs} \neg \varphi) \land G \land \phi \models \bot$ (logical inconsistency) holds.
    Because $Dom \land (\bigwedge_{\varphi \in \BCs} \neg \varphi) \land G_{-i} \land \phi \to Dom \land G_{-i} \land \phi$, $\forall 1 \leq i \leq |G|, Dom \land (\bigwedge_{\varphi \in \BCs} \neg \varphi) \land G_{-i} \land \phi \not \models \bot$ (minimality) holds.
    The non-triviality obviously holds.
\end{proof}

Theorem~\ref{thm:isOriBC} shows that although the additional constraint is considered, the results are still BCs under the original domain properties and goals.

\begin{thm}\label{thm:isNotWit}
    In Algorithm~\ref{alg:JAc}, $\nexists \varphi \in \CBCs$ \St~$\varphi$ is a witness of $\phi$.
\end{thm}

\begin{proof}[Sketch of proof]
    We prove Theorem~\ref{thm:isNotWit} by inductive hypothesis as follows.
    \begin{itemize}
        \item At the first iteration where $\CBCs$ is an empty set, assume we get a BC $\varphi_1$, Theorem~\ref{thm:isNotWit} holds.
        \item We suppose that at the $k$-th iteration where we get a BC $\varphi_k$, Theorem~\ref{thm:isNotWit} holds. 
        \item At the $k$+$1$-th iteration where $\CBCs = \{ \varphi_1,\dots,\varphi_k \}$, assume we get a BC $\phi$.
        Because $\phi$ is a BC under $Dom \cup \{ \neg \varphi | \varphi \in \CBCs \}$ and $G$, $\forall 1 \leq i \leq |G|, Dom \land (\bigwedge_{\varphi \in \CBCs} \neg \varphi) \land G_{-i} \land \phi \not \models \bot$. 
        Therefore, for every $\varphi_j \in \CBCs$, $\phi \land \neg \varphi_j$ is a BC under $Dom \cup \{ \neg \varphi | \varphi \in \CBCs \land \varphi \neq \varphi_j \}$ and $G$. 
        Because of Theorem~\ref{thm:isOriBC}, $\phi \land \neg \varphi_j$ is a BC under $Dom$ and $G$.
    \end{itemize}
\end{proof}

Intuitively, based on Theorem~\ref{thm:isNotWit}, \JAc can produce a search bias towards the BCs that capture different divergences. 
 
\begin{thm}\label{thm:conBC}
    In Algorithm~\ref{alg:JAc}, the BCs in the final $\CBCs$ are not witnesses with each other.
\end{thm}

    

It is straightforward to prove Theorem~\ref{thm:conBC} because of Theorem~\ref{thm:isNotWit} and Algorithm~\ref{alg:incontFil}.
Theorem~\ref{thm:conBC} guarantees that Algorithm~\ref{alg:JAc} returns a set of contrastive BCs.

\input{Tab_exa}

\section{Experiments}\label{sec:expriments}

In this section, we conducted extensive experiments on a broad range of benchmarks shown in Table~\ref{tab:case} to evaluate the performance of \JAc by comparing with \PPAc.
We first presented the research questions.

\begin{que}\label{RQ2}
    What is the performance of the joint framework (\JAc) for producing the contrastive BC set compared with the post-processing framework (\PPAc)?
\end{que}

\subsection{Experimental Setups}

The experimental setups used in this section was the same as the one described in Section~\ref{sec:exp-adv}.
In addition, we added the new experimental setups.
\begin{itemize}
    \item We set the same BC solver (\GA~\cite{degiovanni2018genetic}) for \PPAc and \JAc.
    \item We invoked Aalta~\cite{li2015sat} to check the BC termination condition.
\end{itemize}

\subsection{Experimental Results}

Table~\ref{tab:PPAcvsJAc} shows the overall performance of \PPAc and \JAc, including the running time of \GA (`GA t.'), the running time of the framework (`t.'), and the number of meeting the BC termination condition (`\#T').
In \JAc, $\BCs$ records all BCs identified during the search.

From the column `$|\CBCs|$', the contrastive BCs obtained by \JAc is slightly less than that obtained by \PPAc. 
This is because \JAc not only considers the contrasty in BC but also considers the BC termination condition where \JAc searches for a set of contrastive BCs that is enough so that there is no BC in the domain properties and goals after avoiding these contrastive BCs.
We also observe that the size of $\BCs$ of \JAc is much smaller than that of \PPAc.
Moreover, for \JAc, the size of $\BCs$ is close to that of $\CBCs$.
These observations show that \JAc produces a strong search bias towards the BCs that are contrastive with the identified BCs.

In \PPAc, the running time of \GA is approximately the same as that of producing a set of contrastive BCs.
And the running time of producing a set of contrastive BCs increases as the number of BCs identified by \GA increases.
In particular, in AAP, MP, and ATM, the running time of producing a set of contrastive BCs is about $1.5$ times that of \GA.
It indicates the drawback of \PPAc, namely, the cost of producing a set of contrastive BCs is proportional to the number of BCs identified by a BC solver.
It is foreseeable that the redundant BCs in $\BCs$ will greatly reduce the efficiency of producing a set of contrastive BCs.

\JAc deals with the drawback of \PPAc, because \JAc uses the identified contrastive BC for pruning during the search process, thereby avoiding searching for the redundant BCs.
The shorter running time for meeting the BC termination condition confirms this conclusion.
If \JAc meets the BC termination condition, \JAc will produce a set of contrastive BCs efficiently.

Particularly, in RP1 and ATM, \JAc is $10$ times faster than \PPAc.
We also observe that if \JAc does not meet the BC termination condition (only TCP), \JAc is slower than \PPAc.
It is reasonable because \JAc additionally checks the BC termination condition after finding a new BC.

Conclusively, \JAc produces the search bias towards contrastive BCs.
In addition, the efficiency of \JAc is not limited to the number of BCs identified by a BC solver.

\section{Related Work}\label{sec:relatedwork}

Inconsistency management, \ie, how to deal with inconsistencies in requirements, has also been the focus of several studies, in particular on the formal side.
Besides the inconsistency management approaches based on the informal or semi-formal methods, such as~\cite{hausmann2002detection,herzig2014conceptual,kamalrudin2009automated,kamalrudin2011improving}, a series of formal approaches~\cite{ellen2014detecting,ernst2012agile,harel2005synthesis,nguyen2014kbre} recently have been proposed, which only focus on logical inconsistency or ontology mismatch.
Another related approach is proposed by Nuseibeh and Russo~\cite{nuseibeh1999using}, which generates the conjunction of ground literals as an explanation for the unsatisfiable specification based on abduction reasoning.
As for consistency checking methods, we have to mention the approach of Harel et al.~\cite{harel2005synthesis}, which identifies inconsistencies between two requirements represented as conditional scenarios. 
Moreover, the work \cite{jureta2010techne,liu2010ontology,mairiza2010constructing} studied the reasoning about conflicts in requirements.
In this paper, we focus on the situations that lead to goal divergences, which are nothing but weak inconsistencies.

Goal-conflict analysis has been widely used as an abstraction for risk analysis in GORE.
It is typically driven by the identify-assess-control cycle, aimed at identifying, assessing and resolving inconsistencies that may obstruct the satisfaction of the expected goals.

In identifying inconsistencies, we have to mention the work on obstacle analysis.
An obstacle, first proposed in~\cite{van2000handling}, is a particular goal conflict, which captures the situation that only one goal is inconsistent with the domain properties.
Alrajeh et al.~\cite{alrajeh2012generating} exploited the model checking technique to generate tracks that violate or satisfy the goals, and then to compute obstacles from these tracks based on the machine learning technique.
Other approaches for obstacle analysis include~\cite{cailliau2012probabilistic,cailliau2014integrating,cailliau2015handling,van2000handling}.
Whereas, as obstacles only capture the inconsistency for single goals, these approaches fail to deal with the situation where multiple goals are conflicting.

In this work, we focus on the other inconsistencies -- boundary condition.
Let us come back to the problem of identifying BCs.
Existing approaches mainly categorize into construct-based approaches and search-based approaches.
For construct-based approaches, Van Lamsweerde et al.~\cite{van1998managing} proposed a pattern-based approach which only returns a BC in a pre-defined limited form.
Degiovanni et al.~\cite{degiovanni2016goal} exploited a tableaux-based approach that generates general BCs but only works on small specifications because tableaux are difficult to be constructed.

For the search-based approach, Degiovanni et al.~\cite{degiovanni2018genetic} presented a genetic algorithm which seeks for BCs and handles specifications that are beyond the scope of previous approaches.
Moreover, Degiovanni et al.~\cite{degiovanni2018genetic} first proposed the concept of generality to assess BCs.
Their work filtered out the less general BCs to reduce the set of BCs.
However, the generality is a coarse-grained assessment metric.

As the number of identified inconsistencies increases, the assessment stage and the resolution stage become very expensive and even impractical.
Recently, the assessment stage in GORE has been widely discussed to prioritize inconsistencies to be resolved and suggest which goals to drive attention to for refinements.
However, some of the work~\cite{van2000handling,alrajeh2012generating,cailliau2012probabilistic,cailliau2014integrating,cailliau2015handling} assume that certain probabilistic information on the domain is provided and analyzes to simpler kinds of inconsistencies (obstacles).

In order to automatically assess BCs, Degiovanni et al.~\cite{degiovanni2018goal} recently have proposed an automated approach to assess how likely conflict is, under an assumption that all events are equally likely.
They estimated the likelihood of BCs by counting how many models satisfy a circumstance captured by a BC.
However, the number of models cannot accurately indicate the likelihood of divergence, because not all the circumstances captured by a BC result in divergence.
In this paper, we discovered the drawbacks and proposed a new metric to avoid evaluation mistakes for the likelihood.

For the resolution of conflicts, Murukannaiah et al.~\cite{murukannaiah2015resolving} resolved the conflicts among stakeholder goals of system-to-be based on the Analysis of Competing Hypotheses technique and argumentation patterns.
Related works on conflict resolution also include~\cite{felfernig2009plausible} which calculates the personalized repairs for the conflicts of requirements with the principle of model-based diagnosis.

However, these approaches presuppose that the conflicts have been already identified and our approach for boundary condition discovery provides a footstone for solving these problems.
Let us recall Example~\ref{example}.
Letier et al.~\cite{letier2001reasoning} resolved the BC by refining the first goal as: the pump is switched on when the water level is high and there is no methane. Formally, $\globally((h \wedge \neg m) \rightarrow \nexttime(p))$.

\section{Conclusion and Future Work}\label{sec:conclusion}

Providing a reasonable set of BCs for assessing and resolving divergences is of great significance both from an economical perspective and an impact on software quality.
In this paper, we have proposed a new metric, contrasty, to deal with the drawbacks caused by the generality metric.
Because BCs are ultimately used for resolving divergences, we argue that the identified BCs should help to assess and resolve divergences.
The contrasty metric mainly distinguishes the difference between BCs from the point of resolving divergences.
Experimental results have shown the advantage of contrasty metric, namely, it filters out the BCs capturing the same divergence.
It helps to avoid costly reworks, \ie, assessing and resolving the same divergence captured by redundant BCs.
In addition, we have designed a joint framework to improve the performance of the post-processing framework.

Future work will extend our contrasty metric to the assessment stage and the resolution stage.

\section*{Acknowledgment}
We thank Fangzhen Lin, Yongmei Liu, Jianwen Li, and Ximing Wen for discussion on the paper and anonymous referees for helpful comments. 

\bibliographystyle{IEEEtranS}
\bibliography{cbc}

\end{document}

%% file: Tab_eva.tex
\begin{table}[t]
    \centering
    \caption{The BCs produced by different metrics}\label{tab:BC}%
    \scalebox{0.9}{
      \begin{tabular}{c|c|l|c|l|l}
      \toprule
      \multirow{2}[4]{*}{Case} & \multicolumn{2}{c|}{GL} & \multicolumn{3}{c}{CL} \\
  \cmidrule{2-6}          & \multicolumn{1}{c}{Rank} & \multicolumn{1}{c|}{BC} & \multicolumn{1}{c}{Rank} & \multicolumn{1}{c}{BC} & \multicolumn{1}{c}{Witness} \\
      \midrule
      \multirow{4}[2]{*}{RP1} & 1     & ${\future( ( ( p  \land ( \globally ( \neg q ) ) ) \ \until ( \future( q \land ( \neg p ) ) )) \lor ( \globally ( p \land ( \globally ( \neg q ) ) ) )) }$      & 1     & ${ (  p  \land ( \globally ( \neg q ) ) ) \lor ( \future (  q  \land ( \neg p ) )  ) }$   & 1,2,3,{\color{red}4} \\
            & 2     & ${ \nexttime ( ( p \land ( \globally ( \neg q ) ) ) \lor ( \future( q \land ( \neg p  ) ) ) ) }$      &       &       &  \\
            & 3     &  ${ ( ( \neg q \ \until ( q \land \neg p ) ) \ \until ( \future( p \land (  \globally ( \neg q ) ) ) )) \lor ( \globally ( \neg q \ \until ( q \land \neg p ) ))  }$     &       &       &  \\
            & 4     &  ${ ( p \land ( \globally ( \neg q ) ) ) \lor ( \future( q \land \neg p ) ) }$     &       &       &  \\
      \midrule
      \multirow{3}[2]{*}{RP2} & 1     & ${ \future( ( p \land ( \neg q \land \neg s ) ) \lor ( q \land \neg r ) ) }$      & 1     & ${ \future( ( q \land \neg r ) \lor ( \nexttime ( p \land ( \neg s \ \until ( \neg q \land \neg s ) ) ) ) ) }$      & {\color{red}1},2,3 \\
            & 2     & ${ \nexttime ( ( p \land ( \neg s \ \until ( \neg q \land \neg s )  ) ) \lor ( q \land \neg r ) ) }$      &       &       &  \\
            & 3     & ${ ( p \land ( \neg s \ \until ( \neg q \land \neg s ) ) )  \lor ( q \land \neg r )}$       &       &       &  \\
      \midrule
      \multirow{5}[2]{*}{Ele} & 1     & ${  \nexttime ( \future( call \land ( \globally ( \neg open ) ) ) ) }$     & 1     & ${\nexttime ( \future( call \land ( \globally \neg open ) ) )  }$     & {\color{red}1},3 \\
            & 2     & ${ (( \future( \neg atfloor \land ( \nexttime open ) ) ) \ \until ( call \land  ( \globally ( \neg open ) ) )) \lor (\globally ( \future( \neg atfloor \land ( \nexttime open )   )))  }$    & 2     & ${ open \ \until ( call \land ( \globally \neg open ) )    }$   & 2,{\color{red}4} \\
            & 3     & ${ \globally ( \neg atfloor \land ( \nexttime call ) )  }$     & 3     & ${  ( call \land ( \globally ( \neg open ) ) ) \lor ( \nexttime ( \globally ( call \land (  \globally \neg open ) ) ) ) }$      & 2,3,{\color{red}5} \\
            & 4     & ${ open \ \until ( call \land ( \globally \neg open ) ) }$      &       &       &  \\
            & 5     & ${ ( call \land ( \globally ( \neg open ) ) ) \lor ( \nexttime ( \globally ( call \land (  \globally ( \neg open ) ) ) ) ) }$      &       &       &  \\
      \midrule
      \multirow{3}[2]{*}{TCP} & 1     & ${ ( delivered \land ( send \land \neg ack ) ) \lor ( send  \land ( ack \land \neg delivered ) ) }$      & 1     & ${ \future( (( send \land \neg ack ) \until ( send \land ( ack \land \neg delivered ) )) \lor (\globally ( send \land \neg ack )) )  }$     & 1,{\color{red}2},3 \\
            & 2     & ${ \future( (( send \land \neg ack ) \until ( send \land ( ack  \land \neg delivered ) )) \lor ( \globally ( send \land \neg ack ) ) ) }$      &       &       &  \\
            & 3     & \tabincell{l}{$(( delivered \land \neg ack ) \until ( send \land ack  \land \neg delivered )) \lor $\\$ ( \globally (delivered \land \neg ack ) )$}      &       &       &  \\
      \midrule
      \multirow{5}[2]{*}{AAP} & 1     & ${ \future(r \land p)  }$    & 1     & ${ \future( r \land p )  }$     & {\color{red}1},2,3,5 \\
            & 2     & ${ ( r \land p ) \lor ( \future( r \land ( \future s ) ) )  }$     & 2     & ${ r \land ( \nexttime ( p ) )  }$     & {\color{red}4},5 \\
            & 3     & ${ ( r \land p ) \lor ( \nexttime ( r \land ( \future q ) ) )  }$     &       &       &  \\
            & 4     &  ${ r \land ( \nexttime ( p ) )  }$    &       &       &  \\
            & 5     &  ${ (( r \land p ) \until ( \globally ( p \land ( \globally ( \neg q ) ) ) )) \lor ( \globally ( r \land p ))   }$   &       &       &  \\
      \midrule
      \multirow{6}[2]{*}{MP} & 1     & $\future ( h \land  \neg m \land  p \land  \nexttime (\neg h \land  \neg p \lor  hw \land  (m \lor  \neg p)))$       & 1     &$ \future ( h \land  m ) $       & 1,{\color{red}2},3,4,5,6 \\
            & 2     & $ \future ( m \land  h)$       &       &       &  \\
            & 3     & $ ( m \land  h  ) \lor  \future ( h \land  ( \nexttime  \neg p  )) $       &       &       &  \\
            & 4     & $ ( m \land  h ) \lor  ( m  \land  ( \nexttime  p ) ) $       &       &       &  \\
            & 5     & $ ( m \land  ( \nexttime  p ) ) \ \until ( \nexttime ( m \land  h ) ) \lor  \globally ( m \land  ( \nexttime  p ))$       &       &       &  \\
            & 6     & $ \globally ( h ) \lor  ( m \land  h ) $       &       &       &  \\
      \midrule
      \multirow{8}[2]{*}{ATM} & 1     & $ \future (( \neg p  \land  ( \nexttime  \neg l ) ) \lor  ( \neg m \land  ( p \land   \neg l )))$       & 1     & $ \future ( ( \neg p \land  ( \nexttime  \neg l ) ) \lor  ( \neg m \land  ( p \land  \neg l ) ) )$       & {\color{red}1},2,3,4,5,6,7,8 \\
            & 2     & $ \future (( \neg p \land  ( m \lor  ( \globally  \neg l  ) ) ) \lor  ( \neg m \land  ( p \land  \neg  l ) ) )$      &       &       &  \\
            & 3     & $ \future ( (( \neg m \land  ( p \land  ( \neg l ) ) ) \ \until ( \neg p \land  ( m \lor  ( \nexttime  \neg l ) ))) \lor  \globally ( \neg m \land  ( p  \land  \neg l ) ) )$      &       &       &  \\
            & 4     & $ \future ( \globally ( (\neg p \ \until ( \neg m \land  ( p \land  \neg l  ) )) \lor  \globally  \neg p ) )$      &       &       &  \\
            & 5     & $ (( \neg m \land  p ) \ \until ( \neg p \land  ( m \lor  ( \nexttime  \neg l ) ) )) \lor  \globally ( \neg m \land  p )$      &       &       &  \\
            & 6     & $ ( \neg p \land  ( m  \lor  ( \nexttime  \neg l ) ) ) \lor  ( \neg m  \land  ( p \land  \neg l ) )$      &       &       &  \\
            & 7     & $ \nexttime ( ( \neg p \land  ( m \lor  ( \nexttime  \neg l ) ) ) \lor  ( \neg m \land  ( p \land  \neg l ) ) )$      &       &       &  \\
            & 8     & $ \globally ( ( ( \neg p \ \until ( \nexttime  l )) \lor  \globally  \neg p  ) \lor  ( \neg m  \land  \neg l ) )$      &       &       &  \\
      \midrule
      \multirow{4}[2]{*}{RRCS} & 1     & $ \future ( ( \future ( cc \land  tc ) ) \lor  ( \nexttime ( go \land  ta ) ) )$      & 1     & $ ( cc \land  tc ) \lor  ( \future ( go \land  ta ) )$      & 1,2,{\color{red}3},4 \\
            & 2     & $ ( \future ( cc \land  tc ) ) \lor  ( go \land  ta )$      &       &       &  \\
            & 3     & $ ( cc \land  tc ) \lor  ( \future ( go \land  ta ) )$      &       &       &  \\
            & 4     & $ ( \future ( cc \land  tc ) ) \lor  ( \globally ( go \land  tc ) )$      &       &       &  \\
      \midrule
      \multirow{4}[2]{*}{Tel} & 1     & $ \future (( ( \true \land  \neg d ) \ \until ( \neg c \land  ( \true \land  \neg d ) ) ) \land  c )$      & 1     & $ \future ( ( \neg d \ \until ( f \land  \neg d ) ) \land  c )$      & 1,2,{\color{red}3},4 \\
            & 2     & $ \future ( ( \neg d \ \until ( o \land  \neg d ) ) \land  c )$      &       &       &  \\
            & 3     & $ \future ( ( \neg d \ \until ( f \land  \neg d ) ) \land  c )$      &       &       &  \\
            & 4     & $ \future ( \nexttime ( ( ( \neg d \ \until ( \neg c \land  \neg d )) \lor  (\globally  \neg d )) \land  c ) )$      &       &       &  \\
      \midrule
      \multirow{5}[2]{*}{RRA} & 1     & $ \future ( ( (( \nexttime ( r1 \land  ( \globally  \neg g1 ))) \ \until ( \nexttime ( g0 \land  g1 ) )) \lor  \globally ( \nexttime ( r1 \land  ( \globally  \neg g1  ) )) 
) \lor  ( r0 \land  ( \globally  g1 ) ) )$      & 1     & $ \future ( ( \nexttime (r1 \land  ( \globally  \neg g1 ) ) ) \lor  ( r0 \land  ( \globally  g1 ) ) )$      & 1,{\color{red}2},3,4,5 \\
            & 2     & $ \future ( ( \nexttime (r1 \land  ( \globally  \neg g1 ) ) ) \lor  ( r0 \land  ( \globally  g1 ) ) )$      &       &       &  \\
            & 3     & $ ( (( \nexttime ( r1 \land  ( \globally  \neg g1 ) ) ) \ \until ( g0 \land  g1 )) \lor  \globally ( \nexttime ( r1 \land  ( \globally  \neg g1 ) ) ) ) \lor  ( r0 \land  ( 
\globally  g1 ) )$      &       &       &  \\
            & 4     & $ \nexttime ( (( \nexttime ( r1 \land  ( \globally  g0 ) ) ) \ \until ( r0 \land  ( \globally  r1  ) )) \lor  \globally ( \nexttime ( r1 \land  ( \globally  g0  ) ) ) )$      &       &       &  \\
            & 5     & \tabincell{l}{$ (( (( \nexttime ( r1 \land  ( \globally  \neg g1 ) ) ) \ \until ( \nexttime ( r1 \land  g1 ) )) \lor  \globally ( \nexttime ( r1 \land  ( \globally  \neg g1 ) ) ) ) \ \until ( r0  \land  ( \globally  g1 ) )) \lor $\\$ \globally ( (( \nexttime ( r1 \land  ( \globally  \neg g1 ) ) ) \ \until  ( \nexttime ( r1 \land  g1 ) )) \lor  \globally ( \nexttime (r1 \land  ( \globally  \neg g1 ) ) ) )$}      &       &       &  \\
      \bottomrule
      \end{tabular}%
    }
\end{table}%
  

%% file: Tab_exa.tex
\begin{table*}[t]
    \centering
    \caption{The overall performance of \PPAc and \JAc}\label{tab:PPAcvsJAc}%
    \scalebox{1}{
      \begin{tabular}{c|rrrrr|rrrrr}
      \toprule
      \multirow{2}[4]{*}{Case} & \multicolumn{5}{c|}{\PPAc}             & \multicolumn{5}{c}{\JAc} \\
  \cmidrule{2-11}          & \multicolumn{1}{c}{$|\BCs|$} & \multicolumn{1}{c}{$|\CBCs|$} & \multicolumn{1}{c}{GA t. (s)} & \multicolumn{1}{c}{t. (s)} & \multicolumn{1}{c|}{\#suc.} & \multicolumn{1}{c}{$|\BCs|$} & \multicolumn{1}{c}{$|\CBCs|$} & \multicolumn{1}{c}{\#T} & \multicolumn{1}{c}{t. (s)} & \multicolumn{1}{c}{\#suc.} \\
      \midrule
      RP1   & 37.1  & 1.2   & 157.4 & 224.53 & 10    & \textbf{1} & \textbf{1} & 10    & \textbf{29.5} & 10 \\
      RP2   & 35.1  & 1.2   & 130.2 & 206   & 10    & \textbf{1.1} & \textbf{1.1} & 10    & \textbf{78.9} & 10 \\
      Ele   & 28    & 2.6   & 45.8  & 88.01 & 10    & \textbf{2.1} & \textbf{2.1} & 10    & \textbf{43.4} & 10 \\
      TCP   & 53.9  & 1.5   & 225.1 & \textbf{308.26} & 10    & \textbf{1.4} & \textbf{1.4} & 0     & 801.6 & 10 \\
      AAP   & 50.3  & 1.8   & 65.3  & 208.64 & 10    & \textbf{1} & \textbf{1} & 10    & \textbf{41.3} & 10 \\
      MP    & 40.7  & 1.4   & 59.3  & 146.02 & 10    & \textbf{1} & \textbf{1} & 10    & \textbf{60.8} & 10 \\
      ATM   & 64.4  & 1.2   & 102.2 & 259.19 & 10    & \textbf{1} & \textbf{1} & 10    & \textbf{25.2} & 10 \\
      RRCS  & 27.9  & 1     & 68.3  & 91.87 & 10    & \textbf{1} & 1     & 10    & \textbf{15} & 10 \\
      Tel   & 36.5  & 1     & 35.3  & 46.53 & 2     & \textbf{1} & 1     & 10    & \textbf{27} & \textbf{10} \\
      LAS   & N/A   & N/A   & N/A   & N/A   & 0     & N/A   & N/A   & 0     & N/A   & 0 \\
      PA    & N/A   & N/A   & N/A   & N/A   & 0     & N/A   & N/A   & 0     & N/A   & 0 \\
      RRA   & 40.571 & 1     & 696.43 & 878.7 & 7     & \textbf{1} & 1     & 10    & \textbf{255.1} & \textbf{10} \\
      SA    & N/A   & N/A   & N/A   & N/A   & 0     & N/A   & N/A   & 0     & N/A   & 0 \\
      LB    & N/A   & N/A   & N/A   & N/A   & 0     & N/A   & N/A   & 0     & N/A   & 0 \\
      LC    & N/A   & N/A   & N/A   & N/A   & 0     & N/A   & N/A   & 0     & N/A   & 0 \\
      AMBA  & N/A   & N/A   & N/A   & N/A   & 0     & N/A   & N/A   & 0     & N/A   & 0 \\
      \bottomrule
      \end{tabular}%
    }
\end{table*}%